\documentclass{article}
\usepackage{amsmath,amsthm,graphicx}
\usepackage{mathptmx}

\def\eps{\lambda}

\def\REG{\textit{\bf REG}}

\newtheorem{theorem}{Theorem}
\newtheorem{lemma}[theorem]{Lemma}
\date{}

\begin{document}
\title{Answers to Questions Formulated in the Paper ``On States Observability in Deterministic Finite Automata''}
\author{Tom\'{a}\v{s} Masopust\\
  \small Faculty of Information Technology, Brno University of Technology\\[-0.8ex]
  \small Bo\v{z}et\v{e}chova 2, Brno 61266, Czech Republic\\
  \small \texttt{masopust@fit.vutbr.cz}
}

\maketitle
  \begin{abstract}
    This paper gives answers to questions formulated as open in the paper ``On State Observability in Deterministic Finite Automata'' by A. Mateescu and Gh. P{\u a}un. Specifically, it demonstrates that for all $k\ge 2$, the families of regular languages acceptable by deterministic finite automata with no more than $k$ semi-observable states, denoted by $\mathcal{T}_k$, are anti-AFL's, and that the family $\mathcal{T}_1$ differs in the closure property under Kleene $+$.
  \end{abstract}

\section{Introduction}
  In 1987, Mateescu and P{\u a}un \cite{mateescu} studied state observability in completely specified deterministic finite automata without non-accessible states. They classified the states of finite automata into three types: observable, semi-observable, and non-observable. A state $q$ is said to be {\em observable} if there exists a string $w$ such that $\delta(q,w)$ is a final state; otherwise, $q$ is said to be {\em non-observable}. Moreover, an observable state $q$ is said to be {\em semi-observable} if there is a state $\delta(q,a)$ that is non-observable, for some input symbol $a$. (In what follows, we assume that each deterministic finite automaton is completely specified without non-accessible states and, therefore, call them simply deterministic finite automata.) They proved that the family of all regular languages accepted by completely specified deterministic finite automata with all states being observable, denoted by $\mathcal{O}$, forms a proper subfamily of the family of regular languages which is not closed under union, intersection, complementation, concatenation, intersection with regular sets, $\eps$-free homomorphism, inverse homomorphism, mirror image, and right and left quotient. However, if all the regular languages are taken over a given alphabet $\Sigma$, where $\Sigma$ is minimal for them, then the family of all such regular languages, denoted by $\mathcal{O}(\Sigma)$, is closed under union, concatenation, and left quotient (the other properties are the same as for $\mathcal{O}$). It is also not hard to see that if the automaton has a non-observable state, then there is an equivalent automaton which has only one non-observable state (by the minimization). Thus, we have the well-known result showing that any regular language is accepted by a deterministic finite automaton with no more than one non-observable state.

  On the other hand, considering the number of semi-observable states gives rise to an infinite hierarchy of families of subregular languages, denoted by $\mathcal{T}_k$, $k\ge 0$. In addition, it is known that for $k\ge 4$, all the families $\mathcal{T}_k$ are anti-AFL's, i.e., they are not closed under union, concatenation, Kleene +, $\eps$-free homomorphism, inverse homomorphism, and intersection with regular sets. However, the properties of (some of) the other language families were left open. Note also that in comparison with $\mathcal{T}_k$, $k\ge 4$, $\mathcal{T}_0$ is closed under Kleene $+$ and, therefore, is not an anti-AFL.

  This paper answers these questions and proves that except for $\mathcal{T}_1$, all the families $\mathcal{T}_k$, $k\ge 2$, are anti-AFL's.

\section{Preliminaries and Definitions}
  In this paper, we assume that the reader is familiar with formal language theory (see \cite{salomaa}). For an alphabet (finite nonempty set) $\Sigma$, $\Sigma^*$ represents the free monoid generated by $\Sigma$, where the unit of $\Sigma^*$ is denoted by $\eps$. Set $\Sigma^+ = \Sigma^*-\{\eps\}$. Let $\REG$ denote the family of all regular languages. For other non-specified notions and notations, the reader is referred to \cite{salomaa}.

  Let $\mathcal{A}=(Q,\Sigma,\delta,q_0,F)$ be a completely specified {\em deterministic finite automaton}, i.e., $Q$ is a finite set of states, $\Sigma$ is an input alphabet, $\delta:Q\times\Sigma\to Q$ is a total (completely specified) transition function, $q_0\in Q$ is the initial state, and $F\subseteq Q$ is a (possibly empty) set of final states. Let $L(\mathcal{A})$ denote the language accepted by $\mathcal{A}$, i.e., $L(\mathcal{A})=\{w\in\Sigma^* : \delta(q_0,w)\in F\}$, where $\delta$ is, as usual, extended to be from $Q\times\Sigma^*$ to $Q$.

  In addition, in this paper we assume that all the states of $Q$ are {\em accessible}, which means that for each $q\in Q$, there exists $w\in\Sigma^*$ such that $\delta(q_0,w)=q$. In what follows, the notion of a deterministic finite automaton (DFA) stands for an automaton which is completely specified, deterministic, and without non-accessible states.

  A state $q\in Q$ is said to be {\em observable} if there exists a string $w\in\Sigma^*$ such that $\delta(q,w)\in F$. Otherwise, $q$ is {\em non-observable}. An observable state $q\in Q$ is said to be {\em semi-observable} if for some $a\in\Sigma$, the state $\delta(q,a)$ is non-observable. A DFA $\mathcal{A}$ is said to be observable if all its states are observable, and the language $L(\mathcal{A})$ is said to be observable, too. Let $\mathcal{O}$ denote the family of all observable regular languages. For a given DFA $\mathcal{A}$, denote by $so(A)$ the number of semi-observable states in $\mathcal{A}$. For each $k\ge 0$, define the language family $\mathcal{T}_k=\{L\in\REG : \textrm{ there is a DFA } \mathcal{A} \textrm{ such that } L=L(\mathcal{A}) \textrm{ and } so(A)\le k\}$.

  An alphabet $\Sigma$ is minimal for a given language $L$ if $L\subseteq\Sigma^*$ and for each $\Sigma'\subset\Sigma$ we have $L-\Sigma'^*\neq\emptyset$. Let $\mathcal{O}(\Sigma)$ and $\REG(\Sigma)$ denote the families of observable regular languages and regular languages, respectively, for which $\Sigma$ is the minimal alphabet. Finally, define the language families $\mathcal{T}_k(\Sigma)$ in an analogous way.

\section{An Overview of Known Results}
  This section presents an overview of the main results proved in \cite{mateescu}. Let $L\subseteq\Sigma^*$ be a regular language, and let $Init(L)=\{w\in\Sigma^* : wy\in L \textrm{ for some } y\in\Sigma^*\}$. Then, the following characterization of observable languages is known.
  \begin{theorem}
    A regular language $L\subseteq\Sigma^*$, where $\Sigma$ is minimal for $L$, is observable if and only if $Init(L)=\Sigma^*$.
  \end{theorem}
  It immediately follows from this theorem that any observable language is infinite. In addition, it is known that if $L\subseteq\Sigma^*$ is finite, then $\Sigma^* - L$ is observable. However, the other inclusion does not hold. The following properties are known.

  \begin{theorem}
    The family $\mathcal{O}$ is closed under Kleene $+$.
  \end{theorem}
  \begin{theorem}
    The family $\mathcal{O}$ is not closed under union, intersection, complementation, concatenation, intersection with regular sets, $\eps$-free homomorphism, inverse homomorphism, mirror image, and right and left quotient.
  \end{theorem}
  \begin{theorem}
    Let $\Sigma$ be an alphabet with at least two symbols. Then, the family $\mathcal{O}(\Sigma)$ is not closed under intersection, complementation, homomorphism, inverse homomorphism, mirror image, and right quotient. On the other hand, it is closed under union, concatenation, and left quotient.
  \end{theorem}

  \begin{theorem}$ $
    \begin{itemize}
      \item $\mathcal{T}_0=\mathcal{O}\cup\{\emptyset\}$,
      \item $\mathcal{T}_k\subset\mathcal{T}_{k+1}$, $k\ge 0$, and
      \item $\bigcup_{k\ge 0} \mathcal{T}_k = \REG$.
    \end{itemize}
  \end{theorem}

  Let $L$ be a regular language. To count the smallest $k$ such that $L\in \mathcal{T}_k$, we can use the following lemma.
  \begin{lemma}
    Given a language $L\in \mathcal{T}_{k} - \mathcal{T}_{k-1}$, the value of $k$ can be obtained algorithmically, by constructing a minimal DFA for $L$.
  \end{lemma}
  The following properties are also known:
  \begin{itemize}
    \item $\mathcal{T}_k$ is not closed under union, for $k\ge 2$.
    \item $\mathcal{T}_k$ is not closed under concatenation, for $k\ge 2$.
    \item $\mathcal{T}_k$ is not closed under $\eps$-free homomorphism, for $k\ge 2$.
    \item $\mathcal{T}_k$ is not closed under intersection with regular sets, for $k\ge 0$.
    \item $\mathcal{T}_k$ is not closed under inverse homomorphism, for $k\ge 2$.
    \item $\mathcal{T}_k$ is not closed under Kleene $+$, for $k\ge 4$.
  \end{itemize}
  Except for Kleene $+$, where it is not solved for the families $\mathcal{T}_k$, $k=1,2,3$, it remains to solve these questions for $\mathcal{T}_1$. This is done in the next section.

\section{Results}
  This section answers the questions formulated as open in \cite{mateescu}. It shows that except for the family $\mathcal{T}_1$, all the families $\mathcal{T}_k$ are anti-AFL's, for all $k\ge 2$.
  \begin{lemma}
    Families $\mathcal{T}_1$ and $\mathcal{T}_1(\Sigma)$ are not closed under union, concatenation, $\eps$-free homomorphism, and inverse homomorphism.
  \end{lemma}
  \begin{proof} $ $

    \noindent
    {\it Union:} Let $L_1,L_2\subseteq\{a,b\}^*$, $L_1=(a+b)a^*$ and $L_2=(a+b)b^*$. Then, $L_1,L_2\in\mathcal{T}_1$ because they are accepted by DFAs $\mathcal{M}_1=(\{s,r,f\},\{a,b\},\delta_1,s,\{f\})$ and $\mathcal{M}_2=(\{s,r,f\},$ $\{a,b\},\delta_2,s,\{f\})$, respectively, where $\delta_1(s,a)=\delta_1(s,b)=f$, $\delta_1(f,a)=f$, $\delta_1(f,b) = \delta_1(r,a) = \delta_1(r,b) = r$, and $\delta_2(s,a)=\delta_2(s,b)=f$, $\delta_2(f,b)=f$, $\delta_2(f,a)=\delta_2(r,a)=\delta_2(r,b)=r$. However, $L_1\cup L_2\notin\mathcal{T}_1$ because the minimal DFA accepting it is $\mathcal{M}=(\{1,2,3,4,5\}, \{a,b\}, \delta, 1, \{2,3,4\})$, where $\delta(1,a)=\delta(1,b)=2$, $\delta(2,a)=3$, $\delta(2,b)=4$, $\delta(3,a)=3$, $\delta(3,b)=5$, $\delta(4,a)=5$, $\delta(4,b)=4$, and $\delta(5,a)=\delta(5,b)=5$. Clearly, states $3$ and $4$ are semi-observable, which implies that $L_1\cup L_2\in \mathcal{T}_2-\mathcal{T}_1$.

    \bigskip\noindent
    {\it Concatenation:} Let $L_1$ and $L_2$ be two languages over $\{a,b\}$ defined by the following two deterministic finite automata $\mathcal{A}_1$ and $\mathcal{A}_2$, respectively (see Fig. \ref{concat1}). As it is easy to complete the automaton, the non-observable states are omitted from now on.
    \begin{figure}[ht]
      \begin{center}
        \includegraphics[scale=.5]{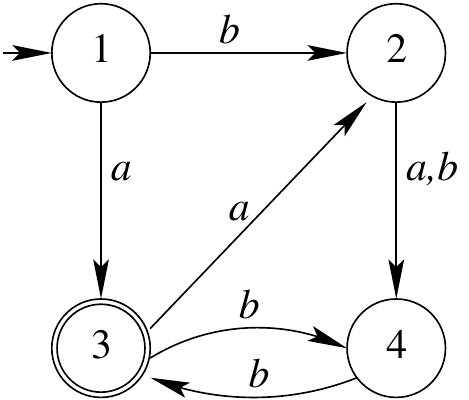} \qquad
        \includegraphics[scale=.5]{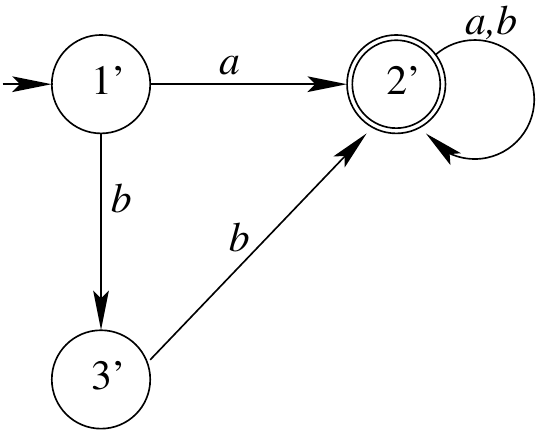}
        \caption{DFAs $A_1$ and $A_2$.}
        \label{concat1}
      \end{center}
    \end{figure}
    Then, $L_1,L_2\in\mathcal{T}_1$, with $\{a,b\}$ being their minimal alphabet, because only states $4$ and $3'$ are semi-observable. As the minimal DFA accepting the concatenation of these two languages $L_1$ and $L_2$ is as in Fig. \ref{concat2},
    \begin{figure}[ht]
      \begin{center}
        \includegraphics[scale=.5]{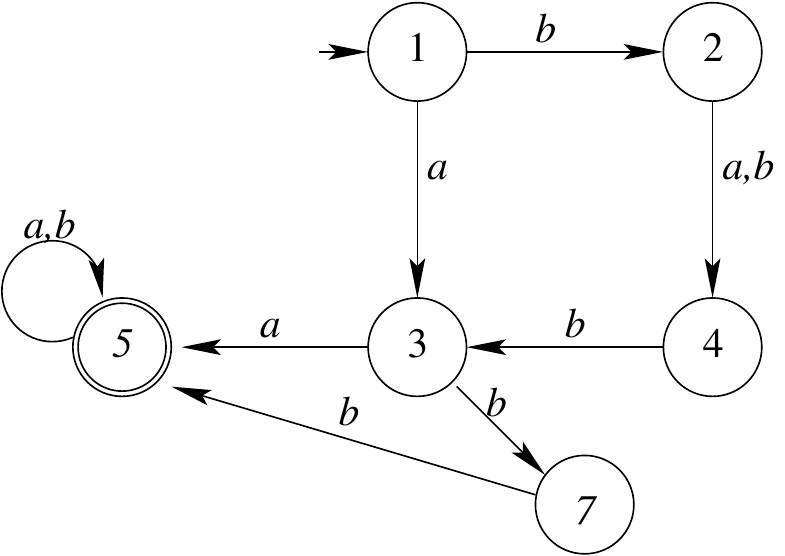}
        \caption{The minimal DFA for $L_1\cdot L_2$ with two semi-observable states, $4$ and $7$.}
        \label{concat2}
      \end{center}
    \end{figure}
     which has two semi-observable states, namely 4 and 7, the language $L_1\cdot L_2 \notin \mathcal{T}_1$.

     Note that considering two languages over different alphabets, this result can be proved more easily. Specifically, let $L_1=a^+$ and $L_2=b^+$. Then, $L_1,L_2\in\mathcal{T}_0$, but $L_1\cdot L_2\notin\mathcal{T}_1$.

    \bigskip\noindent
    {\it $\eps$-free homomorphism:} Let $L\subseteq\{a,b\}^*$ be a language defined by the following deterministic finite automaton $\mathcal{A}$ (see Fig. \ref{homo}),
    \begin{figure}
      \begin{center}
        \includegraphics[scale=.5]{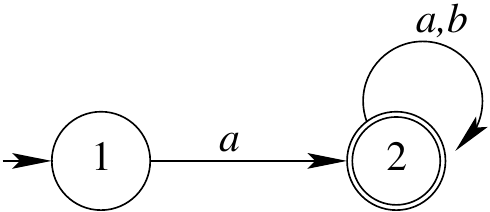} \qquad
        \includegraphics[scale=.5]{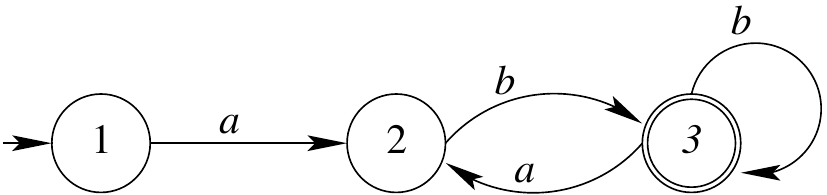}
        \caption{DFAs for $L$ and $h(L)$, respectively.}
        \label{homo}
      \end{center}
    \end{figure}
    and let $h:\{a,b\}^*\to\{a,b\}^*$ be a homomorphism defined as $h(a)=ab$, $h(b)=b$. Then, $L\in\mathcal{T}_1$, but $h(L)\notin\mathcal{T}_1$ because the minimal DFA accepting $h(L)$ has two semi-observable states (see Fig. \ref{homo}).

    \bigskip\noindent
    {\it Inverse homomorphism:} Let $L\subseteq\{a,b\}^*$ be a language defined by the following deterministic finite automaton $\mathcal{A}$ (see Fig. \ref{inv}), and let $h:\{a,b\}^*\to\{a,b\}^*$ be a homomorphism defined as $h(a)=aba$ and $h(b)=bab$. Then, $L\in\mathcal{T}_1$, but $h^{-1}(L)\notin\mathcal{T}_1$. See Fig. \ref{inv} for minimal DFAs accepting $L$ and $h^{-1}(L)$.
    \begin{figure}[ht]
      \begin{center}
        \includegraphics[scale=.5]{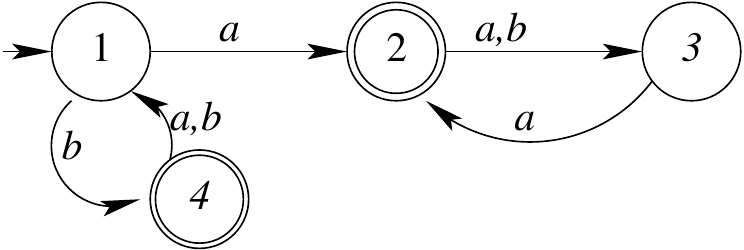} \qquad
        \includegraphics[scale=.5]{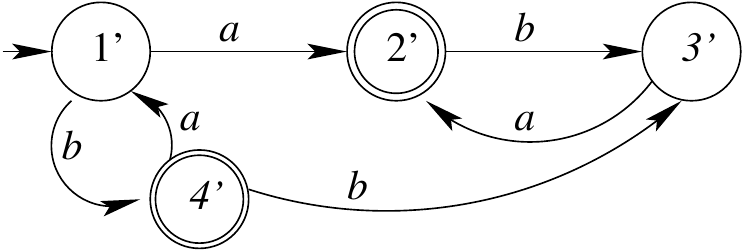}
        \caption{DFAs for $L$ and $h^{-1}(L)$, respectively.}
        \label{inv}
      \end{center}
    \end{figure}
  \end{proof}

  \begin{lemma}
    Let $\Sigma$ be an alphabet. The families $\mathcal{T}_2$, $\mathcal{T}_3$, $\mathcal{T}_2(\Sigma)$, and $\mathcal{T}_3(\Sigma)$ are not closed under Kleene $+$.
  \end{lemma}
  \begin{proof}
    Let $L\subseteq\{a,b\}^*$ be a language defined by the following deterministic finite automaton $\mathcal{A}$ (see Fig. \ref{kleene}).
    \begin{figure}[ht]
      \begin{center}
        \includegraphics[scale=.5]{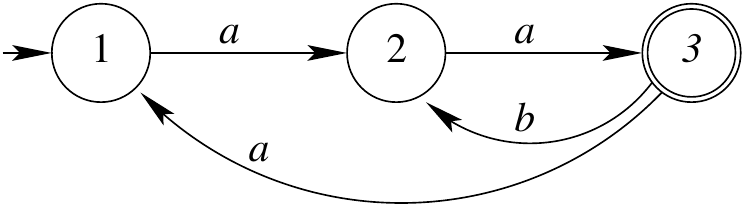} \qquad
        \includegraphics[scale=.5]{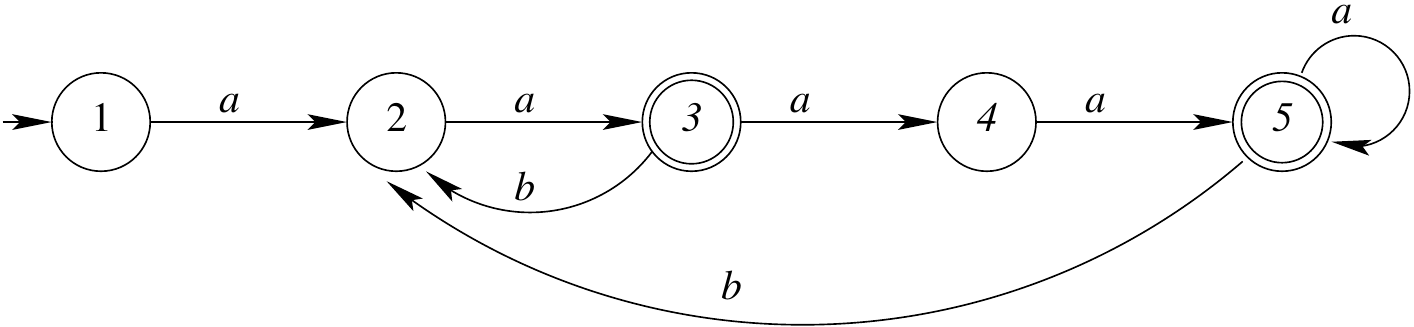}
        \caption{DFAs for $L$ and $L^+$, respectively.}
        \label{kleene}
      \end{center}
    \end{figure}
    Then, $L\in\mathcal{T}_2$, but $L^+\notin\mathcal{T}_2$. The proof for $\mathcal{T}_3$ and $\mathcal{T}_3(\Sigma)$ is analogous.
  \end{proof}

  \begin{lemma}
    Families $\mathcal{T}_1$ and $\mathcal{T}_1(\Sigma)$ are closed under Kleene $+$.
  \end{lemma}
  \begin{proof}
    Let $L\in\mathcal{T}_1(\Sigma)$, and let $\mathcal{M}=(Q,\Sigma,\delta,q_0,F)$ be a minimal deterministic finite automaton such that $L(\mathcal{M})=L$. Moreover, let $q\in Q$ be the only semi-observable state of $\mathcal{M}$. Apply the following algorithm constructing a deterministic finite automaton $\mathcal{N}$ such that $L(\mathcal{N})=L^+$:
    \begin{enumerate}
      \item Add a $\eps$-transition from each final state to the initial state.
      \item Use the common algorithm for removing $\eps$-transitions.
      \item Use the subset construction to construct a deterministic automaton $\mathcal{N'}$.
      \item Minimize the automaton $\mathcal{N'}$ and remove all non-accessible states.
    \end{enumerate}
    Clearly, $\mathcal{N'}=(2^Q,\Sigma,\delta',\{q_0\},\{X\in 2^Q : X\cap F\neq\emptyset\})$. Let $\mathcal{N}$ denote the final automaton, i.e., $\mathcal{N'}$ where all non-accessible states (including the adjacent transitions) are removed.

    Let $a\in\Sigma$ be such that $k=\delta(q,a)$ is non-observable in $\mathcal{M}$. Then, if $\{q\}$ is a state of $\mathcal{N}$, then also $\{k\}=\delta'(\{q\},a)$ is a state of $\mathcal{N}$. In addition, because $\mathcal{M}$ is minimal, $\delta(k,a)=k$ for all $a\in\Sigma$, and we have that $\{k\}$ is a non-observable state. Thus, $\{q\}$ is a semi-observable state of $\mathcal{N}$.

    Assume that $\emptyset\neq X\subseteq Q$ is a state of $\mathcal{N}$ different from $\{q\}$ (clearly, $\emptyset$ is non-accessible). Let $p\in X$ such that $p\neq q$. As $p$ is observable which is not semi-observable in $\mathcal{M}$, there exists $w_a\in\Sigma^*$, for each $a\in\Sigma$, such that $\delta(p,aw_a)\in F$. Thus, $\delta'(X,aw_a)\cap F\neq\emptyset$ and, therefore, $X$ is an observable state which is not semi-observable.
  \end{proof}

  The arguments of the previous proof also explain why the families $\mathcal{T}_k$, $k\ge 2$, are not closed under Kleene $+$. Let $p$ and $q$ be semi-observable states of a deterministic finite automaton accepting $L$. Then, the deterministic finite automaton accepting $L^+$ constructed as in the previous proof can have more semi-observable states than the original one, namely $\{p\}$, $\{q\}$, and $\{p,q\}$.

\subsection*{Acknowledgements}
  This work was supported by the Czech Ministry of Education under the Research Plan No. MSM~0021630528.

\end{document}